\documentclass[11pt]{article}
\usepackage{fullpage}

\usepackage{graphics}
\usepackage[dvips]{epsfig}

\usepackage{amsmath}
\usepackage{amssymb}
\usepackage{amsfonts}
\usepackage{graphicx}

\usepackage{algorithm}
\usepackage{algpseudocode}
\usepackage{amsthm}
\usepackage{color}
\usepackage{xspace}

\newtheorem{theorem}{Theorem}[section]

\newtheorem{claim}{Claim}[section]
\newtheorem{proposition}{Proposition}[section]

\newenvironment{proofclaim}{\noindent{\bf Proof of the claim.}}{\hfill$\star$}

\newcommand{\cA}{{\cal A}}

\newcommand{\remove}[1]{}

\newcommand{\cP}{{\cal P}}

\newcommand{\ie}{{\em i.e.,}\xspace}
\newcommand{\eg}{{\em e.g.,}\xspace}

\begin{document}

	\baselineskip  0.21in 
	\parskip     0.1in 
	\parindent   0.0in 

	\title{{\bf On Deterministic Rendezvous at a Node of Agents\\ with Arbitrary Velocities}}

	\author{
		S\'{e}bastien Bouchard\thanks{Sorbonne Universit\'{e}s, UPMC Universit\'{e} Paris 06, CNRS, INRIA, LIP6 UMR 7606, Paris, France, E-mail: sebastien.bouchard@lip6.fr}
		\and
		Yoann Dieudonn\'{e}\thanks{
		MIS Lab., Universit\'{e} de Picardie Jules Verne, France, E-mail: yoann.dieudonne@u-picardie.fr}
		\and
		Andrzej Pelc\thanks{D\'{e}partement d'informatique, Universit\'{e} du Qu\'{e}bec en Outaouais,
		Gatineau, Qu\'{e}bec J8X 3X7,
		Canada. E-mail: pelc@uqo.ca.
		Supported in part by NSERC discovery grant 8136 -- 2013
		and by the Research Chair in Distributed Computing of
		the Universit\'{e} du Qu\'{e}bec en Outaouais.}
		\and
		Franck Petit\thanks{Sorbonne Universit\'{e}s, UPMC Universit\'{e} Paris 06, CNRS, INRIA, LIP6 UMR 7606, Paris, France, E-mail: Franck.Petit@lip6.fr}
	}

	\date{ }

	\maketitle

	\begin{abstract}
		We consider the task of rendezvous in networks modeled as undirected graphs. Two mobile agents with different labels, starting at different nodes of an anonymous graph, have to meet. This task has been considered in the literature under two alternative scenarios: {\em weak} and {\em strong}. Under the weak scenario, agents may meet either at a node or inside an edge. Under the strong scenario, they have to meet at a node, and they do not even notice meetings inside an edge. Rendezvous algorithms under the strong scenario are known for synchronous agents. For asynchronous agents, rendezvous under the strong scenario is impossible even in the two-node graph, and hence only algorithms under the weak scenario were constructed. In this paper we show that rendezvous under the strong scenario is possible for agents with asynchrony restricted in the following way: agents have the same measure of time but the adversary can impose, for each agent and each edge, the speed of traversing this edge by this agent. The speeds may be different for different edges and different agents but all traversals of a given edge by a given agent have to be at the same imposed speed. We construct a deterministic rendezvous algorithm for such agents, working in time polynomial in the size of the graph, in the length of the smaller label, and in the largest edge traversal time.

		\vspace{2ex}

	\end{abstract}

	\section{Introduction}

		{\bf The background.}
			We consider the task of rendezvous in networks modeled as undirected graphs. Two mobile entities, called agents,  have different positive integer labels, start from different nodes of the network, and have to meet. Mobile entities may represent software agents in a communication network, or physical mobile robots, if the network is a labyrinth or a cave, or if it consists of corridors of a building.
		The reason to meet may be to exchange previously collected data or ground or air samples, or to split work in a future task of network exploration or maintenance.

			The task of rendezvous in networks has been considered in the literature under two alternative scenarios: {\em weak} and {\em strong}. Under the weak scenario \cite{CLP,DGKKP,DPV}, agents may meet either	at a node or inside an edge. Under the strong scenario \cite{DFKP,KM,TSZ07}, they have to meet at a node, and they do not even notice meetings inside an edge. Each of these scenarios is appropriate in different applications. The weak scenario is suitable for physical robots in a network of corridors, while the strong scenario is needed for software agents in computer networks.

			Rendezvous algorithms under the strong scenario are known for synchronous agents, where time is slotted in rounds, and in each round each agent can either wait at a node or move to an adjacent node. For asynchronous agents, where an agent decides to which neighbor it wants to move but the adversary totally controls the walk of the agent and can arbitrarily vary its speed, rendezvous under the strong scenario is impossible even in the two-node graph, and hence only algorithms under the weak scenario were constructed.

			However, due to the fact that the strong scenario is appropriate for software agents in computer networks, and that such agents are rarely synchronous, it is important to design rendezvous algorithms under the strong scenario, restricting the asynchrony of the agents as little as possible. This is the aim of this paper. We consider mobile agents with asynchrony restricted as follows: agents have the same measure of time but the adversary can impose, for each agent and each edge, the speed of traversing this edge by this agent. The speeds may be different for different edges and different agents but all traversals of a given edge by a given agent have to be at the same imposed speed. We are interested in deterministic rendezvous algorithms for such agents.

		{\bf The model.}\\
			{\em The network.} The network is modeled as a simple undirected connected graph. As in the majority of papers on rendezvous, we seek algorithms that do not rely on the knowledge of node labels, and we assume that the underlying graph is anonymous. Designing such algorithms is important because even when node labels exist, nodes may refuse to reveal them, \eg due to security or privacy reasons. It should be also noted that, if nodes had distinct labels, agents might explore the graph and meet in the smallest node, hence gathering would reduce to graph exploration. On the other hand, we make the assumption, again standard in the literature of the domain, that edges incident to a node $v$ have distinct labels in $\{0,\dots,d-1\}$, where $d$ is the degree of $v$. Thus every undirected edge $\{u,v\}$ has two labels, which are called its {\em port numbers} at $u$ and at $v$. Port numbers are visible to the agents. Port numbering is {\em local}, \ie there is no relation between port numbers at $u$ and at $v$. Note that in the absence of port numbers, edges incident to a node could not be distinguished by agents and thus rendezvous would be often impossible, as the adversary could prevent an agent from traversing some edge incident to the current node. Also, the above mentioned concerns of security and privacy that may prevent nodes from revealing their labels, do not apply to port numbers.

			{\em The agents.} Agents $A_1$ and $A_2$ start at arbitrary different nodes of the graph. They are placed at their starting nodes at the beginning. They cannot mark visited nodes or traversed edges in any way. The adversary wakes up the agents at possibly different times. Agents do not  know the topology of the graph nor any bound on its size. They have clocks ticking at the same rate. The clock of each agent starts at its wakeup, and at this time the agent starts executing the algorithm. When an agent enters a node, it learns its degree and the port of entry. We assume that the memory of the agents is unlimited: from the computational point of view they are modeled as Turing machines.

			{\em The adversary.} The adversary assigns different positive integer labels to both agents. Each agent knows a priori only its own label. Both agents execute the same deterministic algorithm whose parameter is the label of the agent. Moreover, for each edge $e$ of the graph, the adversary assigns two positive reals: $t_1(e)$ and $t_2(e)$. During the execution of an algorithm, an agent can wait at the currently visited node for a time of its choice, or it may choose a port to traverse the corresponding edge $e$. In the latter case, agent $A_r$ traverses this edge in time $t_r(e)$,  getting to the other end of the edge after this time. This modelling permits a lot of asynchrony: agents can have different velocities when traversing different edges, and an agent slower in one edge can be faster in another edge. This is motivated by the fact that congestion and bandwidth of different edges may be different, and that each of the agents can have a different traversing priority level on different edges. In particular, this general scenario includes the model of agents walking at possibly different constant velocities, that was used in \cite{DiPe} for the task of approach in the plane.

			The {\em time} of a rendezvous algorithm is the worst-case time between the wakeup of the earlier agent and the meeting at a node.

			{\bf Our results.} We construct a deterministic rendezvous algorithm working for arbitrary graphs under the strong scenario. Our algorithm works in time polynomial in $n$, $\ell$ and $\tau$, where $n$ is the number of nodes of the graph, $\ell$ is the logarithm (\ie the length) of the smaller label, and $\tau$ is the maximum of all values $t_1(e)$ and $t_2(e)$  assigned by the adversary, over all edges $e$ of the graph.

		{\bf Related work.}
			A survey of  randomized rendezvous in various scenarios  can be found in \cite{alpern02b}. Deterministic rendezvous in networks was surveyed in \cite{Pe}. In many papers rendezvous was considered in a geometric setting: an interval of the real line, see, \eg \cite{baston01,gal99}, or the plane, see, \eg \cite{anderson98a,anderson98b}).

			For deterministic rendezvous in networks, attention concentrated on the study of the feasibility of rendezvous, and on the time required to achieve this task, when feasible. For example, deterministic rendezvous with agents equipped with tokens used to mark nodes was considered, \eg in~\cite{KKSS}. Deterministic rendezvous of two agents that cannot mark nodes but have unique labels was discussed in \cite{DFKP,KM,TSZ07}. All these papers were concerned with the time of rendezvous in arbitrary graphs. In \cite{DFKP} the authors showed a rendezvous algorithm polynomial in the size of the graph, in the length of the shorter label and in the delay between the starting times of the agents. In \cite{KM,TSZ07} rendezvous time was polynomial in the first two of these parameters and independent of the delay. All the above papers assumed that agents are synchronous, and used the scenario called {\em strong} in the present paper.

			Several authors investigated asynchronous rendezvous in the plane \cite{CFPS,fpsw} and in network environments \cite{CLP,DGKKP,DPV}. In the latter scenario, the agent chooses the edge which it decides to traverse but the adversary totally controls the walk of the agent inside the edge and can arbitrarily vary its speed. Under this assumption, rendezvous under the strong scenario cannot be guaranteed even in very simple graphs, and hence the rendezvous requirement was weakened by considering the scenario called {\em weak} in the present paper. In particular, the main result of \cite{DPV} is an asynchronous rendezvous algorithm working in an arbitrary graph at cost (measured by the number of edge traversals) polynomial in the size of the graph and in the logarithm of the smaller label. The scenario of possibly different fixed speeds of the agents was introduced in \cite{DiPe}.

	\section{Preliminaries}\label{prelim}

		Throughout the paper, the number of nodes of a graph is called its size. Procedure ${\tt Explo}$, based on universal exploration sequences (UXS), is a corollary of the result of Reingold \cite{Re}. Given any positive integer $n$, it allows the agent to visit all nodes of any graph of size at most $n$, starting from any node of this graph and coming back to it, using $P(n)$ edge traversals, where $P$ is some {increasing} polynomial. In the first half of the procedure, after entering a node of degree $d$ by some port $p$, the agent can compute the port $q$ by which it has to exit; more precisely $q=(p+x_i)\mod d$, where $x_i$ is the corresponding term of the UXS. In the second half of the procedure, the agent backtracks to its starting node.

		We will use the following terminology. The agent woken up earlier by the adversary is called the {\em earlier} agent and the other agent is called the {\em later} agent. If agents are woken up simultaneously, these appellations are given arbitrarily. Consider executions $E_1$ and $E_2$, respectively of procedures $\cP _1$ and  $\cP _2$ by agents $A_1$ and $A_2$. Executions $E_1$ and $E_2$ are called {\em overlapping}, if the time segments that they occupy are not disjoint.

		We define the following transformation of labels. Consider a label $x$ of an agent, with binary representation $(c_1\dots c_r)$. Define the {\em modified label} of the agent to be the sequence $M(x)=(c_1c_1c_2c_2\dots c_rc_r01)$. Note that, for any $x$ and $y$, the sequence $M(x)$ is never a prefix of $M(y)$. Also, $M(x) \neq M(y)$ for $x\neq y$.

	\section{The algorithm and its analysis}

		The strong rendezvous procedure is called ${\tt Strong\-RV}$ (shown in Algorithm~\ref{alg:ipl:rv}) and its execution requires to call procedure ${\tt Phase}(h)$ that is described in Algorithm~\ref{alg:ipl:phase}. At a high level, ${\tt Phase}(h)$ consists of executions of ${\tt Explo}(h)$ and carefully scheduled waiting periods of various lengths, designed according to the bits of the modified label of the agent. The aim is to guarantee a period in which one agent stays still at a node and the other visits all nodes of the graph.

		\begin{algorithm}[H]
			\caption{Algorithm~${\tt Strong\-RV}$} \label{alg:ipl:rv}
			\begin{algorithmic}[1]
				\State $h\gets 1$
				\While{agents have not met}
					\State execute ${\tt Phase}(h)$
					\State $h\gets 2h$
				\EndWhile
				\State declare that the gathering is achieved
			\end{algorithmic}
		\end{algorithm}

		\begin{algorithm}[H]
			\caption{${\tt Phase}(h)$} \label{alg:ipl:phase}
			\begin{algorithmic}[1]
				\State Let $x$ be the label of the agent and let $M(x)=(b_1b_2\ldots b_s)$, where $s$ is the length of $M(x)$.
				\State /* Initialization */
				\State execute ${\tt Explo}(h)$ \label{line:ipl:init:1}
				\State wait for time $4h^2 (\sum_{j=0}^{\log(h)} (P(2^j)))$ \label{line:ipl:init:2}
				\State /* Core */
				\State $i \gets 1$ \label{line:ipl:core:1}
				\While {$i \leq h$}
					\If {$i>s$ or $b_i=0$}
						\State wait for time $2hP(h)$
						\State execute twice ${\tt Explo}(h)$
					\Else
						\State execute twice ${\tt Explo}(h)$
						\State wait for time $2hP(h)$
					\EndIf
					\State $i \gets i + 1$
				\EndWhile \label{line:ipl:core:2}
				\State /* End */
				\State Wait for time $hP(2h)$ \label{line:ipl:end:1}
				\State Execute ${\tt Explo}(h)$
			\end{algorithmic}
		\end{algorithm}

		The correctness and time complexity of Algorithm~${\tt Strong\-RV}$ are now analyzed. In the following statements and proofs, $\alpha$ denotes the smallest power of two which upper bounds the following three numbers: the size $n\geq 2$ of the graph $G$, the length of the smaller modified label $2\ell+2$, and the parameter $\tau$ the maximum of all traversal durations assigned by the adversary over all edges of $G$.

		The following proposition directly follows from the definitions of $\alpha$ and UXS.

		\begin{proposition}
			\label{proof:ipl:wait:1}
			For any positive integers $x$ and $y$ such that $x \geq \alpha$, $xP(y)$ upper bounds the time required by any agent to execute ${\tt Explo}(y)$ in $G$.
		\end{proposition}

		\begin{proposition}
			\label{proof:ipl:wait:2}
			For any positive integer $x$ and any power of two $y$ such that $x \geq \alpha$ and $x \geq y$, $T_{x, y} = 4xy \sum_{z=0}^{\log(y)} P(2^z)$ upper bounds the time required by any agent to execute the sequence $S_y = {\tt Phase}(1), {\tt Phase}(2), \dots, {\tt Phase}(\frac{y}{4}), {\tt Phase}(\frac{y}{2}), {\tt Explo}(y)$ in graph $G$.
		\end{proposition}

		\begin{proof}
			Let an arbitrary positive integer at least $\alpha$ be assigned to $x$. The proof is made by induction on $y$. Consider the case where $y=1$. In this case, the sequence $S_y$ consists only of ${\tt Explo}(1)$. In view of Proposition~\ref{proof:ipl:wait:1}, $T_{x, 1}$, which is equal to $4xP(1)$, upper bounds the time required by any agent to execute ${\tt Explo}(1)$, which proves the first step of the induction. Now, assume that there exists a power of two $1\leq z\leq x$, such that the statement of the proposition holds for $y=z$. The next paragraph proves that if $2z\leq x$, then the statement of the lemma holds also for $2z$.

			Denote by ${\tt Suffix}(z)$ the sequence of instructions of ${\tt Phase}(z)$ deprived from the first call to ${\tt Explo}(z)$. The sequence $S_{2z}$ is successively made of $S_z$, ${\tt Suffix}(z)$ and ${\tt Explo}(2z)$. By the inductive hypothesis, the time required to execute $S_z$ is upper bounded by $T_{x, z}$. By Proposition~\ref{proof:ipl:wait:1}, the time required to execute ${\tt Explo}(2z)$ is upper bounded by $xP(2z)$. In view of Algorithm~\ref{alg:ipl:phase} and Proposition~\ref{proof:ipl:wait:1}, the time required to execute ${\tt Suffix}(z)$ is upper bounded by $T_{x, z} + 2z(x+z)P(z) + zP(2z) + xP(z)$. Hence, the maximal duration of $S_{2z}$ is upper bounded by $2T_{x, z} + (x+z)P(2z) + P(z)(2z(x+z) + x)$, which is at most $2T_{x, z} + (2z(x+z)+2x+z) P(2z)$ as $P(2z) \geq P(z)$. Moreover, $2T_{x, z} + (2z(x+z)+2x+z) P(2z) \leq 2T_{x, z} + 4x(2z)P(2z) = T_{x, 2z}$. This shows that the statement of the proposition also holds when $y=2z$, which proves the proposition.
		\end{proof}

		The following theorem proves the correctness of Algorithm~${\tt Strong\-RV}$.

		\begin{theorem}
			\label{proof:ipl:theo1}
			Algorithm~${\tt Strong\-RV}$ guarantees rendezvous in $G$ by the time the first of the agents completes the execution of ${\tt Phase}(2\alpha)$.
		\end{theorem}

		\begin{proof}
			Assume by contradiction that the statement of the theorem is false. Note that when any agent finishes the first execution of ${\tt Explo}(\alpha)$ of ${\tt Phase}(\alpha)$ (line~\ref{line:ipl:init:1} of Algorithm~\ref{alg:ipl:phase}) it has visited every node of $G$, and thus the other agent has been woken up before the end of this execution, or else the agents would have met.

			The core of ${\tt Phase}(\alpha)$ (lines \ref{line:ipl:core:1}-\ref{line:ipl:core:2} of Algorithm~\ref{alg:ipl:phase}) can be viewed as a sequence of $\alpha$ blocks, where the $x$-th block (for $1\leq x\leq \alpha$) corresponds to processing bit $b_x$ of the modified label if $x\leq s$ and bit $0$ otherwise. Each of these blocks in turn can be viewed as a sequence of $4$ sub-blocks, each of which corresponds either to a waiting period of length $\alpha P(\alpha)$, or to a single execution of ${\tt Explo}(\alpha)$. Let $I_1,I_2,\ldots,I_{4\alpha}$ (resp. $J_1,J_2,\ldots,J_{4\alpha}$) be the sequence of the $4\alpha$ sub-blocks executed by agent $A$ (resp. agent $B$) in the core of ${\tt Phase}(\alpha)$. The proof of this theorem relies on the following claim.

			\begin{claim}
				For every $1\leq y \leq 4\alpha$, $I_y$ and $J_y$ are concurrent.
			\end{claim}

			\begin{proofclaim}
				Assume by contradiction that $y=x$ is the smallest integer for which it does not hold. Without loss of generality, suppose that the first agent to complete its $x$-th sub-block is $A$. If $x=1$, then in view of Proposition~\ref{proof:ipl:wait:2}, when $A$ starts and finishes $I_1$, $A_2$ is executing the first waiting period of ${\tt Phase}(\alpha)$. Since $I_1$ corresponds to ${\tt Explo}(\alpha)$, as the first bit of a modified label is always $1$, a meeting occurs by the end of the execution of $I_1$ because $\alpha\geq n$, which is a contradiction. So, $x>1$ and since $x$ is minimal, $I_{x-1}$ and $J_{x-1}$ are concurrent. So, when $A$ starts $I_x$, $B$ is executing $J_{x-1}$ and when $A$ completes $I_x$, the execution of $J_{x-1}$ has not yet been completed. This implies that the time required by $A$ to execute $I_x$ is shorter than the time required by $B$ to execute $J_{x-1}$. Hence, $I_x$ cannot be a sub-block corresponding to a waiting period, as each of these periods has length $\alpha P(\alpha)$, which upper bounds the duration of every sub-block corresponding to ${\tt Explo}(\alpha)$ (in view of Proposition~\ref{proof:ipl:wait:1}). Thus $I_x$ corresponds to an execution of ${\tt Explo}(\alpha)$, and so does $J_{x-1}$, as otherwise rendezvous would occur by the time $I_x$ is completed, which would be a contradiction.

				Consider the time lag between executions of $I_x$ and $J_x$. Let $\theta_1=t_B-t_A$, where $t_A$ (resp. $t_B$) is the time when $A$ (resp. $B$) starts $I_x$ (resp. $J_x$). The time required by $A$ (resp. $B$) to execute ${\tt Explo}(\alpha)$ never changes because it is always executed from the initial node of $A$ (resp. $B$), and is at most $\theta_1$ (resp. at least $\theta_1$). Moreover, in each block there are always four sub-blocks: either two waiting periods of $\alpha P(\alpha)$ followed by two ${\tt Explo}(\alpha)$, or vice versa. Consider each of the four positions that can be occupied by $I_x$ in its corresponding block. If it is the first, second, or fourth sub-block of its block, then since $I_x$ and $J_{x-1}$ correspond to ${\tt Explo}(\alpha)$, the number of whole sub-blocks corresponding to ${\tt Explo}(\alpha)$ (resp. the waiting period of $\alpha P(\alpha)$) that remain to be executed by $A$ from $t_A$ is the same as the number of those that remain to be executed by $B$ from $t_A$. If $I_x$ is the third sub-block of its block, then $J_x$ and $J_{x+1}$ correspond to the waiting period while $I_x$ and $I_{x+1}$ correspond to the execution of ${\tt Explo}(\alpha)$. In this case, the number of whole blocks that remain to be executed by $A$ from $t_A$ is the same as the number of those that remain to be executed by $B$ from $t_A$. Moreover, in view of Proposition~\ref{proof:ipl:wait:1}, the time needed by $B$ to execute $J_x$ and $J_{x+1}$ is at least the time needed by $A$ to execute $I_x$ and $I_{x+1}$.

				As a result, whichever the position of $I_x$ in its block, $A$ is the first agent to finish the core of ${\tt Phase}(\alpha)$ and there exists a difference of $\theta_2\geq\theta_1$ between the times when $A$ and $B$ complete this core. To conclude the proof of the claim, two cases are considered: either $\theta_2$ is longer than the time $A$ needs to execute ${\tt Explo}(2\alpha)$, or not.

				In the first case, since $A$ executes ${\tt Explo}(\alpha)$ faster than $B$, it necessarily completes the end of ${\tt Phase}(\alpha)$ at least time $\theta_2$ ahead of $B$. As a consequence, it starts executing ${\tt Phase}(2\alpha)$, and in particular its first instruction ${\tt Explo}(2\alpha)$, at least time $\theta_2$ ahead of $B$. This implies that $A$ completes the execution of the first instruction ${\tt Explo}(2\alpha)$ of ${\tt Phase}(2\alpha)$ before $B$ starts it. Hence, $A$ starts the execution of the first waiting period of ${\tt Phase}(2\alpha)$ by the time $B$ starts the first execution of ${\tt Explo}(2\alpha)$ in ${\tt Phase}(2\alpha)$. In view of Proposition~\ref{proof:ipl:wait:2}, this leads to a meeting before any agent starts the core of ${\tt Phase}(2\alpha)$, which is a contradiction. In the second case, $A$ completes the last waiting period of ${\tt Phase}(\alpha)$  at a time $\theta_2$ ahead of $B$. Moreover, $\theta_2$ is at most the duration of this waiting period and at least the time required by $A$ to execute ${\tt Explo}(\alpha)$. Hence, while $A$ executes entirely the last instruction ${\tt Explo}(\alpha)$ of ${\tt Phase}(\alpha)$, $B$ is waiting (it executes the last waiting period of ${\tt Phase}(\alpha)$). This leads to a meeting before any agent starts ${\tt Phase}(2\alpha)$, and hence the second case also results in a contradiction, which proves the claim.
			\end{proofclaim}

			Note that the length of the smaller modified label is at most $\alpha$. Moreover, the modified labels are not prefixes of each other, hence they must differ at some bit. Thus, it follows from the claim that there exists a period during which an agent is waiting in ${\tt Phase}(\alpha)$ while the other entirely executes ${\tt Explo}(\alpha)$. Hence a meeting occurs before any agent starts ${\tt Phase}(2\alpha)$, which is a contradiction and proves the theorem.
		\end{proof}

		According to Theorem~\ref{proof:ipl:theo1}, rendezvous occurs by the time the first of the two agents completes ${\tt Phase}(2\alpha)$, which occurs, by Proposition~\ref{proof:ipl:wait:2}, before this agent has spent at most a time $T_{4\alpha, 4\alpha}$ since its wake up. However, in view of the definition of $\alpha$ and of the modified labels, $T_{4\alpha, 4\alpha}$ is polynomial in $\alpha$ and, thus, in $n$, $\ell$ and $\tau$. This proves the following theorem.

		\begin{theorem}
			\label{proof:ipl:theo2}
			The execution time of Algorithm~${\tt Strong\-RV}$ is polynomial in $n$, $\ell$ and $\tau$.
		\end{theorem}

	\section{Discussion of alternative scenarios}

		Our result shows that the time of rendezvous can be polynomial in $n$, $\ell$ and $\tau$, where $n$ is the number of nodes of the graph, $\ell$ is the logarithm of the smaller label, and $\tau$ is the maximum of all values $t_1(e)$, $t_2(e)$ assigned by the adversary, over all edges $e$ of the graph. We assumed that the time is counted since the wakeup of the earlier agent. It is natural to ask, if it is possible to construct a rendezvous algorithm whose time depends on $n$, $\ell$ and $\tau'$, where $\tau'$ is the {\em minimum} of all values $t_1(e)$, $t_2(e)$ assigned by the adversary, over all edges $e$ of the graph. The answer is trivially negative, if time is counted, as we do, since the wakeup of the earlier agent. Indeed, suppose that there exists such an algorithm working in some time $F(n,\ell,\tau')$. The adversary assigns $t_1(e)>F(n,\ell,\tau')$, for the first edge $e$ taken by the first agent, starts the first agent at some time $t_0$ and delays the wakeup of the second agent until time $t_0+t_1(e)$. Rendezvous cannot happen before time $t_0+t_1(e)$, which is a contradiction.

		It turns out that the answer is also negative in the easier scenario, when time is counted since the wakeup of the agent that is woken up later.Consider even a simplified situation, where $t_1=t_1(e)$ is the same for all edges $e$, and $t_2=t_2(e)$ is the same for all edges $e$. In other words, each of the agents has a constant speed. Thus $\tau=\max(t_1,t_2)$ and $\tau'=\min(t_1,t_2)$. Call the agent $A_r$ for which $t_r$ is larger the {\em slower} agent, and the other -- the {\em faster} agent.

		First notice that, if we show that any rendezvous algorithm must take time at least $\tau$ since the wakeup of the later agent, the negative result follows, as the adversary can assign $\tau>F(n,\ell,\tau')$. We now show that, indeed, for any rendezvous algorithm, there exists a behavior of the adversary for which this algorithm takes time at least $\tau$ since the wakeup of the later agent.

		Denote by $\beta$ (resp. by $\gamma$) the waiting time between the wakeup of the faster (resp. slower) agent and the time when it starts the first edge traversal. Let $d=|\beta -\gamma |$.

		If $\beta \geq \gamma$, the adversary wakes up the faster agent at some time $t_0$ and wakes up the slower agent at time $t_0+d$. Both agents start traversing their first edge at the same time $t_0+\beta$ and cannot meet before time $t_0+\beta+\tau$. Since both agents were awake at time $t_0+\beta$, our claim follows.

		If $\beta< \gamma$, the adversary wakes up the slower agent at some time $t_0$ and wakes up the faster agent at time $t_0+d$. Both agents start traversing their first edge at the same time $t_0+\gamma$ and cannot meet before time $t_0+\gamma+\tau$. Since both agents were awake at time $t_0+\gamma$, our claim follows. This concludes the justification that it is impossible to guarantee rendezvous in time depending on $n$, $\ell$ and $\tau'$, even when time is counted since the wakeup of the later agent.

		The above remark holds under the strong scenario considered in this paper. By contrast, the answer to the same question turns out to be positive in the weak scenario. This can be justified as follows. In \cite{DPV} the authors showed a rendezvous algorithm with cost (measured by the total number of edge traversals) polynomial in $n$ and $\ell$ under the weak scenario, assuming that agents are totally asynchronous. Let $\cA$ be this algorithm and let its cost be $K(n,\ell)$, where $K$ is some polynomial. Consider algorithm $\cA$ in the simplified model mentioned above, where each of the agents has a constant speed, the speeds being possibly different, still under the weak scenario. Consider the part of the cost of the algorithm counted since the wakeup of the later agent. This cost is $K'(n,\ell)$, where $K'$ is some polynomial. Of course, the behavior of the agents  in which an agent never waits at a node and crosses each edge at its constant speed is a possible behavior imposed by a totally asynchronous adversary, and hence the result of \cite{DPV} still holds (under the weak scenario). Hence, if time is counted from the wakeup of the later agent, the time of the algorithm from \cite{DPV} is at most $\tau'\cdot K'(n,\ell)$, and therefore it is polynomial in $n$, $\ell$ and $\tau'$.

		Hence, we get the following observation that shows a provable difference between the time of rendezvous under the strong and the weak scenarios, even in the situation when rendezvous is possible under both of these scenarios. If time is counted from the wakeup of the later agent, and agents have constant, possibly different velocities, then rendezvous in time depending on $n$, $\ell$ and $\tau'$ cannot be guaranteed under the strong scenario, but there is a rendezvous algorithm working in time polynomial in $n$, $\ell$ and $\tau'$ under the weak scenario.

		We conclude the paper by stating the following problem. What is the strongest adversary under which strong rendezvous in arbitrary graphs is possible? We showed that this is the case under an adversary that imposes possibly different speeds for different agents and different edges, but the speed must be the same for all traversals of a given edge by a given agent. On the other hand,  it is easy to see that if the adversary can impose a different speed {\em for each traversal} of each edge by each agent (thus the speeds may vary for different traversals of the same edge by the same agent) then strong rendezvous is impossible even in the two-node graph.

		\vspace{-0.2cm}

	\bibliographystyle{plain}

\begin{thebibliography}{12}




	\bibitem{alpern02b}
	S. Alpern and S. Gal,
	The theory of search games and rendezvous.
	Int. Series in Operations research and Management Science,
	Kluwer Academic Publisher, 2002.





	\bibitem{anderson98a}
	E. Anderson and S. Fekete,
	Asymmetric rendezvous on the plane,
	Proc. 14th Annual ACM Symp. on Computational Geometry (1998), 365-373.

	\bibitem{anderson98b}
	E. Anderson and S. Fekete,
	Two-dimensional rendezvous search,
	Operations Research 49 (2001), 107-118.


	\bibitem{baston01}
	V. Baston and S. Gal,
	Rendezvous search when marks are left at the starting
	points,
	Naval Reaserch Logistics 48 (2001), 722-731.



	\bibitem{CFPS}
	M. Cieliebak, P. Flocchini, G. Prencipe, N. Santoro,
	Distributed Computing by Mobile Robots: Gathering. SIAM J. Comput. 41(2012), 829-879.


	\bibitem{CLP}
	J. Czyzowicz, A. Labourel, A. Pelc, How to meet asynchronously (almost) everywhere, ACM Transactions on Algorithms 8 (2012), article 37.

	\bibitem{DGKKP}
	G. De Marco, L. Gargano, E. Kranakis, D. Krizanc, A. Pelc, U. Vaccaro,
	Asynchronous deterministic rendezvous in graphs,
	Theoretical Computer Science 355 (2006), 315-326.

	\bibitem{DFKP}
	A. Dessmark, P. Fraigniaud, D. Kowalski, A. Pelc.
	Deterministic rendezvous in graphs.
	Algorithmica 46 (2006), 69-96.

	\bibitem{DiPe}
	Y. Dieudonn\'{e}, A. Pelc, Deterministic polynomial approach in the plane, Distributed Computing 28 (2015), 111-129.



	\bibitem{DPV}
	Y. Dieudonn\'{e}, A. Pelc, V. Villain, How to meet asynchronously at polynomial cost, SIAM Journal on Computing 44 (2015), 844-867.


	\bibitem{fpsw}
	P. Flocchini, G. Prencipe, N. Santoro, P. Widmayer,
	Gathering of asynchronous robots with limited visibility. Theor. Comput. Sci. 337 (2005), 147-168.




	\bibitem{gal99}
	S. Gal,
	Rendezvous search on the line,
	Operations Research 47 (1999), 974-976.




	\bibitem{KM}
	D. Kowalski, A. Malinowski,
	How to meet in anonymous network,
	 Theor. Comput. Sci. 399 (2008), 141-156.


	\bibitem{KKSS}
	E. Kranakis, D. Krizanc, N. Santoro and C. Sawchuk,
	Mobile agent rendezvous in a ring,
	Proc. 23rd Int. Conference on Distributed Computing Systems
	(ICDCS 2003), IEEE, 592-599.




	\bibitem{Pe}
	A. Pelc, Deterministic rendezvous in networks: A comprehensive survey,
	Networks 59 (2012), 331-34.

	\bibitem{Re}
	O. Reingold, Undirected connectivity in log-space, Journal of the ACM 55 (2008), 1-24.



	\bibitem{TSZ07}
	A. Ta-Shma and U. Zwick.
	Deterministic rendezvous, treasure hunts and strongly universal exploration sequences.
	ACM Trans. Algorithms 10 (2014): 12:1-12:15.




	\end{thebibliography}
	

\end{document}